\documentclass[11pt,a4paper]{article}
\usepackage[left=1.6cm,right=1.6cm,top=1.5cm,bottom=1.5cm]{geometry}
\usepackage{amsmath} 
\usepackage{amsthm}
\usepackage{amssymb}
\usepackage{mathtools} 
\usepackage{mathtools}
\usepackage{float}
\usepackage{subfigure,color}

\newtheorem{theorem}{Theorem}

\newtheorem{lemma}{Lemma}
\newtheorem{corollary}{Corollary}

\newtheorem{assumption}{Assumption}

\providecommand{\keywords}[1]
{
	\small	
	\textbf{\textit{Keywords---}} #1
}

\title{Distributed Learning with Partial Information Sharing}
\date{}
\author{
  P Raghavendra Rao\\
  Department of Electrical Engineering\\
  IIT Tirupati\\
  \texttt{ee20d002@iittp.ac.in}
  \and
  Pooja Vyavahare\\
  Department of Electrical Engineering\\
  IIT Tirupati\\
  \texttt{poojav@iittp.ac.in}
}

\begin{document}

\maketitle

\begin{abstract}
	This work studies the distributed learning process on a network of agents. Agents make partial observation about an unknown hypothesis and iteratively share their beliefs over a set of possible hypotheses with their neighbors to learn the true hypothesis. We present and analyze a distributed learning algorithm in which agents share belief on only one randomly chosen hypothesis at a time. Agents estimate the beliefs on missed hypotheses using previously shared beliefs. We show that agents learn the true hypothesis almost surely under standard network connectivity and observation model assumptions if belief on each hypothesis is shared with positive probability at every time. We also present a memory-efficient variant of the learning algorithm with partial belief sharing and present simulation results to compare rate of convergence of full and partial information sharing algorithms.  
\end{abstract}


\section{Introduction}
\label{sec:intro}

Distributed hypothesis testing (or distributed learning) over a network models opinion dynamics with applications in varied fields like social learning, control systems etc. In this setting agents observe private signals generated based on an unknown hypothesis, which are partially informative, and communicate with each other over a network to learn the true hypothesis. Non-Bayesian distributed learning algorithms in which agents share their \textit{belief vector}\footnote{Belief vector of an agent is a probability vector representing its confidence on a hypothesis being the true hypothesis.} on the set of possible hypotheses has been proposed in the literature \cite{Lalitha18, Mitra21a}. 

While most of the works \cite{Lalitha18, Mitra21a} are based on agents sharing the complete belief vector with  neighbors at all time, in recent years communication efficient algorithms in which agents share belief on only one (or a subset of) hypothesis has gained attention \cite{Mitra21b, Kayaalp24, Toghani22}. Authors in \cite{Mitra21b} modify the learning algorithm of \cite{Mitra21a} to enable quantized and event triggered communication. While \cite{Toghani22} present and analyze variation of algorithm of \cite{Lalitha18} for various quantized communication schemes, authors in \cite{Kayaalp24} present the variation with limited communication (by sharing belief on one hypothesis at a time). In this work, we present two variations of algorithm of \cite{Mitra21a} which enable agents to share belief on only one hypothesis at a time and show that agents learn the true hypothesis almost surely.

\section{System model}
\label{sec:model}

We consider a discrete-time system in which a set of $N$ agents are connected over a strongly connected network\footnote{A strongly connected network is a graph in which there is a path between any two pair of agents.} $G = (V,E).$ Here $V$ is the set of $N$ agents and $E$ is the set of communication edges between them. An agent $j$ is said to be a neighbor of agent $i$ if $(i,j) \in E.$ Let $N_i$ be the set of neighbors of agent $i.$
All agents are trying to find the one true hypothesis $h^*$ from a set of $M$ hypotheses $\mathcal{H} = \{h_1,\ldots, h_M\}.$ Every agent $i \in V$ makes a private observation $o_{i,t} \in \mathcal{O}_i$ at time $t \in \mathbb{N}^+$ where the observation $o_{i,t}$ is generated based on the conditional likelihood function $f_i(.|h^*).$ We assume that the observations of an agent are i.i.d. over time with $|\mathcal{O}_i| < \infty$ and independent across agents. Each agent knows its own likelihood function $f_i(o|h) >0,~\forall o \in \mathcal{O}_i, ~\forall h \in \mathcal{H}$ but the private observations give only partial information about the hypothesis $h^*.$ 

The challenge in this setting is to learn the true hypothesis $h^*$ by iteratively communicating with the neighbors. To do so, at every time $t$ each agent $i$ maintains a probability vector, $\beta_{i,t},$ on $\mathcal{H}$ where $\beta_{i,t}(h)$ represents the agents $i$'s belief (confidence) on hypothesis $h$ being the true hypothesis at time $t.$ Every agent starts with equal belief on each hypothesis, i.e., $\beta_{i,0}(h) = 1/M,~\forall h \in \mathcal{H}.$ We say that the \textit{true learning} happens in the network when $\beta_{i,t}(h^*) \xrightarrow{\text{a.s.}} 1, ~\forall i \in V.$
\subsection{Learning with full information sharing}
\label{sec:full}

In this work, we focus on the \textit{min-rule} based non-Bayesian distributed learning (introduced in \cite{Mitra21a}) in which agents share their public belief vector $\beta_{i,t}$ with neighbors. In \cite{Mitra21a}, at every time $t,$ each agent $i$ observes $o_{i,t}$ and updates a \textit{local} belief vector to get an intermediate belief vector as:
\begin{equation}
\alpha_{i,t}(h) = \frac{f_i(o_{i,t}|h) \alpha_{i,t-1}(h)}{\sum_{h' \in \mathcal{H}} f_i(o_{i,t}|h') \alpha_{i,t-1}(h')}. \label{eq:local_belief}
\end{equation} 
Every agent starts with $\alpha_{i,0}(h) = 1/M,~\forall h \in \mathcal{H}.$ Then, every agent collects the belief vector $\beta_{j,t-1}$ of all its neighbors $j$ and updates its own belief vector as:
\begin{equation}
\beta_{i,t}(h) = \frac{\min\left( \left(\beta_{j,t-1}(h)\right)_{ j \in N_i}, \beta_{i,t-1}(h), \alpha_{i,t}(h) \right)}{\sum_{h' \in \mathcal{H}}\min\left( \left(\beta_{j,t-1}(h')\right)_{ j \in N_i}, \beta_{i,t-1}(h'),\alpha_{i,t}(h') \right)}. \label{eq:public_full}
\end{equation}
The following standard assumption enables agents to learn the true hypothesis by interacting with each other and Theorem~1 of \cite{Mitra21a} show that the true learning happens in the network under some network connectivity conditions.
\begin{assumption}\label{as:global}
	For every pair $h_l,h_k \in \mathcal{H}$ such that $h_l \neq h_k,$ there exists at least one agent $i \in V$ which can distinguish between the hypotheses $h_l, h_k.$ An agent $i$ can distinguish between $h_l, h_k$ if the KL-divergence between the corresponding likelihood distributions is strictly positive, i.e.,
	\begin{equation*}
	K_i(h_l, h_k) = \sum_{o_i \in \mathcal{O}_i} f_i(o_i|h_l) \log \frac{f_i(o_i|h_l) }{f_i(o_i|h_k)} > 0. \label{eq:KL}
	\end{equation*}
\end{assumption}
The set of all agents that can distinguish between the hypotheses $h_l$ and $h_k$ is denoted by $D(h_l,h_k).$
\begin{figure}
	\begin{center}
	\includegraphics[width=0.8\textwidth]{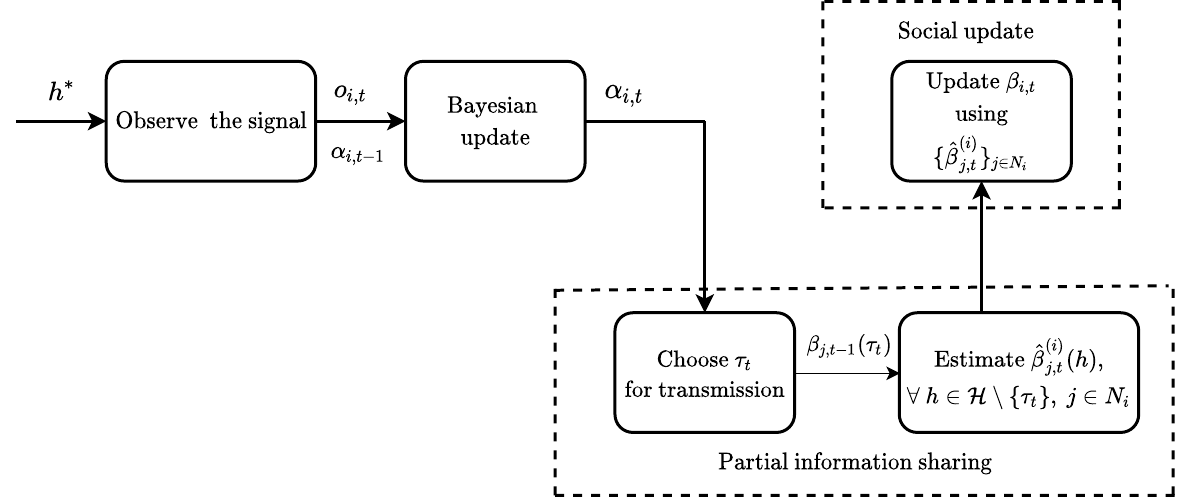}
	\end{center}
\caption{Distributed learning algorithm at agent $i \in V$ with partial information sharing.}
\label{fig:sysmodel}
\end{figure}
%
%
Note that for the belief update in \eqref{eq:public_full} every agent needs to share the $M$-dimensional belief vector to all its neighbors at every time. Motivated by the social learning setting, in which agents tend to share opinion on the trending topic at any time, in the next section we present distributed learning rules in which agents share belief on only one hypothesis at a time.

\section{Learning with partial information sharing}
\label{sec:partial}

In this section, we present a learning rule which allows agents to share belief on only one randomly chosen hypothesis at a time. At every time $t,$ each agent $i$ observes $o_{i,t}$ and updates its local belief vector $\alpha_{i,t}$ by \eqref{eq:local_belief}. Then, every agent collects the belief  $\beta_{j,t-1}(\tau_t)$ from all its neighbors $j \in N_i$ where $\tau_t \in \mathcal{H}$ is chosen uniformly at random. We assume that $\tau_t$ is i.i.d. across time and all agents pick same hypothesis $\tau_t$ for belief sharing at time $t.$ However, our results hold true even when every agent $i$ chooses hypothesis $\tau_t^i$ at time $t$ independent of every other agent in the network. 

Note that every agent $i$ has only one belief value $\beta_{j,t-1}(\tau_t)$ of its neighbor $j$ and needs to estimate $j$'s belief on other hypotheses which it does by maintaining an estimate vector of $j$'s belief at every time. Let $\hat{\beta}_{j,t}^{(i)}$ be the estimate of $\beta_{j,t}$ vector at agent $i$ with initial value as $ \hat{\beta}_{j,0}^{(i)}(h) := 1/M,~\forall h \in \mathcal{H},i \in V,j \in N_i.$ At any time $t,$ agent $i$ updates the estimate for agent $j$ as follows:
\begin{align}
\hat{\beta}_{j,t}^{(i)}(h) = \begin{cases}
\beta_{j,t-1}(h) /D, ~~ h = \tau_t \\
\hat{\beta}_{j,t-1}^{(i)}(h)/D, ~~ h \neq \tau_t.
\end{cases} \label{eq:estimate_previous}
\end{align}
Here $D = 1- \hat{\beta}_{j,t-1}^{(i)}(\tau_t) + \beta_{j,t-1}(\tau_t)$ is to ensure that the estimate vector $\hat{\beta}_{j,t}^{(i)}$ is a probability vector. Agent $i$ uses these estimates to update its public belief (similar to \eqref{eq:public_full}) as follows:
\begin{equation}
\beta_{i,t}(h) = \frac{\min\left( \left(\hat{\beta}_{j,t}^{(i)}(h)\right)_{ j \in N_i}, \beta_{i,t-1}(h), \alpha_{i,t}(h) \right)}{\sum\limits_{h' \in \mathcal{H}}\min\left( \left(\hat{\beta}_{j,t}^{(i)}(h')\right)_{ j \in N_i}, \beta_{i,t-1}(h'),\alpha_{i,t}(h') \right)}. \label{eq:public_previous}
\end{equation}
We show that all the agents eventually learn the true hypothesis and reject all false hypotheses as presented below.
\begin{theorem}\label{th:partial_previous}
	Let the communication graph $G$ be strongly connected, and observation model satisfies Assumption~\ref{as:global}. If every agent shares belief on only one randomly chosen hypothesis at a time and uses estimates from \eqref{eq:estimate_previous} to update beliefs in \eqref{eq:local_belief}, \eqref{eq:public_previous}, then the true learning happens with rate of rejection of a false hypothesis, $h,$ being lower bounded by maximum KL-divergence across $G$ between the pair of hypothesis $h^{*}$ and $h$ as: 
		\begin{equation*}
			\liminf_{t \rightarrow \infty} -\frac{\log{\beta_{i,t}(h)}}{t} \geq \max\limits_{j \in V}K_j(h^{*},h) ~~\text{a.s.}
		\end{equation*} 
\end{theorem}

\subsection{Proof of Theorem~\ref{th:partial_previous}}
\label{sec:results}

In this section, we present the proof of Theorem~\ref{th:partial_previous}. Let the probability space generated by the sequence of $o_{i,t}$ and $\tau_t$ be $(\Omega, \mathcal{F}, \mathsf{P})$ where $\Omega = \{\omega: \omega = (o_{1,t},\ldots,o_{N,t}, \tau_t)~ o_{i,t} \in \mathcal{O}_i, \tau_t \in \mathcal{H}, t \in \mathbb{N}^+)\}$ be the sample space, $\mathcal{F}$ be the $\sigma$-algebra generated by observations and $\tau_t$ and $\mathsf{P}$ is the probability measure induced by sample paths in $\Omega.$ 

Observe that, for any agent $i \in V,$ the local belief $\alpha_{i,t}(h)$ on a hypothesis $h$ depends only on its own private observations $o_{i,t}$ and likelihood function $f_{i}(.|h)$ as it is updated using Bayesian rule in \eqref{eq:local_belief}. Thus, for any agent $i$ which can distinguish between the true hypothesis $h^*$ and a false hypothesis $h,$ one can show that $\alpha_{i,t}(h)$ goes to zero almost surely. More formally,

\begin{lemma} \cite{Mitra21a} \label{lm:localbelief}
	Consider a false hypothesis $h \in \mathcal{H},$ and an agent $i \in D(h^{*},h).$ Then update rule \eqref{eq:local_belief} ensures that:
	\begin{enumerate}
		\item $\alpha_{i,t}(h) \xrightarrow{\text{a.s.}}{} 0,$
		\item $\alpha_{i,\infty}(h^{*}) \triangleq \lim\limits_{t \rightarrow \infty}\alpha_{i,t}(h^{*})$ \text{~exists a.s., and~} \\ $\alpha_{i,\infty}(h^{*}) \geq \alpha_{i,0}(h^{*}),$ and
		\item $\frac{1}{t}\log\frac{\alpha_{i,t}(h)}{\alpha_{i,t}(h^{*})} \xrightarrow{\text{a.s.}} -K_i(h^*,h).$ 
	\end{enumerate}
\end{lemma}

Let $\Omega_0$ be the set of sample paths for which the result of Lemma~\ref{lm:localbelief} holds. The following lemma shows that for any sample path $\omega \in \Omega_0,$ there is a finite time $t'(\omega)$ such that local and public belief on true hypothesis for every agent is bounded away from zero after time $t'(\omega).$

\begin{lemma}  \label{lm:bound_true_hypo}
	For every sample path $\omega \in \Omega_0,$ with update rules \eqref{eq:local_belief}, \eqref{eq:estimate_previous} and \eqref{eq:public_previous}, there exists constants $\nu(\omega) \in (0,1)$ and $t'(\omega) \in (0,\infty)$ such that the following holds true:
	\begin{enumerate}
		\item $\alpha_{i,t}(h^{*})\geq \nu(\omega), ~ \forall ~ t \geq t'(\omega), ~\forall ~ i \in V$ and
		\item $\beta_{i,t}(h^{*})\geq \nu(\omega), ~ \forall ~ t \geq t'(\omega), ~\forall ~ i \in V.$
	\end{enumerate}
\end{lemma}
\begin{proof}
	We prove the result for a sample path $\omega \in \Omega_0$ and omit mentioning $\omega$ where contextually clear. Let $\eta_1 := \min\limits_ {\forall i \in V}\alpha_{i,0}(h^{*})>0$ and there exists $\rho>0$ such that $\eta_1 - \rho >0.$ By employing arguments similar to those used in the proof of part (1) of Lemma 2 in \cite{Rao24}, one can establish that for each sample path $\omega,$ there exists a time-step $t'(\omega)$ such that for all $t \geq t'(\omega),$  $\alpha_{i,t}(h^{*}) \geq \eta_1 - \rho > 0,~\forall i \in V.$ To prove part $(ii)$ of the Lemma, Let $\eta_2(t) := \min\limits_ {\forall i \in V}\beta_{i,t}(h^{*})$ and $\eta_3(t) := \min\limits_ {\forall i \in V, ~j \in N_i} \hat{\beta}_{j,t}^{(i)}(h^{*}).$ We claim that $\eta_2(t)>0$ and $\eta_3(t)>0.$ We will prove this claim by induction. Using the update rules \eqref{eq:local_belief}, \eqref{eq:estimate_previous} and \eqref{eq:public_previous}, it is straightforward to verify that $\eta_2(1)>0$ and $\eta_3(1)>0$ due to non-zero prior beliefs. This proves the base case of the induction argument. To prove the induction step, let $\eta_2(t'-1)>0$ and $\eta_3(t'-1)>0.$ These terms can only be zero if any agent $i \in V$ sets its local belief on $h^{*}$ to zero which is not possible based on the bound on $\alpha_{i,t}.$ It is clear from \eqref{eq:estimate_previous} that $\eta_3(t')>0$ because $\eta_3(t'-1)>0$ and $\eta_2(t'-1)>0.$ Let $\nu=\min \{\eta_1-\rho,\eta_2(t'-1),\eta_3(t')\}>0.$ Then, for an agent $i \in V,$ the public belief on $h^{*}$ at time $t'$ is:
	\begin{align*}
	\beta_{i,t'}(h^{*}) &= \frac{\min\left( \beta_{i,t'-1}(h^{*}),\left(\hat{\beta}_{j,t'}^{(i)}(h^{*})\right)_{j \in N_i}, \alpha_{i,t'}(h^{*}) \right)}{\sum\limits_{k=1}^{M}\min\left( \beta_{i,t'-1}(h_k),\left(\hat{\beta}_{j,t'}^{(i)}(h_k)\right)_{j \in N_i}, \alpha_{i,t'}(h_k) \right)} \\
	&\stackrel {\text{(a)}} {\geq} \frac{\nu}{\sum\limits_{k=1}^{M}\min\left( \beta_{i,t'-1}(h_k),\left(\hat{\beta}_{j,t'}^{(i)}(h_k)\right)_{j \in N_i}, \alpha_{i,t'}(h_k) \right)} \\
	&\stackrel {\text{(b)}} {\geq} \frac{\nu}{\sum\limits_{k=1}^{M} \alpha_{i,t'}(\theta_k)} \stackrel {\text{(c)}}{=}\nu.
	\end{align*}
	Here inequality (a) is derived from the definition of $\nu$ and the condition $\alpha_{i,t}(h^{*}) \geq \eta_1-\rho,~\forall t \geq t'.$ Inequality (b) is obtained by applying upper bound to the denominator terms. Step (c) utilizes the fact that $\alpha_{i,t'}(.)$ is a probability vector. Consequently, by induction, $\beta_{i,t}(h^{*}) \geq \nu,~\forall t > t'.$
\end{proof}
The following lemma proves a bound on the rate of rejection of a false hypothesis by a discriminating agent.

\begin{lemma} \label{lm:rate_discriminating}
	Consider any false hypothesis $h \in \mathcal{H} \setminus \{h^{*}\},$ and let $i \in D(h^{*},h)$ be an agent. Then, the following holds true:
	\begin{equation*}
	\lim_{t \rightarrow \infty} \inf~ -\frac{\log{\beta_{i,t}(h)}}{t} \geq K_i(h^{*},h) ~~\text{a.s.}
	\end{equation*}
\end{lemma}
\begin{proof}
	We demonstrate the result for each sample path $\omega \in \Omega_0,$ omitting explicit mention of $\omega$ when the context is clear. Consider any false hypothesis $h \in \mathcal{H} \setminus \{h^{*}\},$ and an agent $i \in D(h^{*},h).$ For any $\epsilon > 0,$ note that since $i \in D(h^{*},h),$  part $(iii)$ in Lemma \ref{lm:localbelief} implies that there exists a time $t_i(\omega, h , \epsilon)$ such that
	\begin{align} \label{eq:local_upper}
	\alpha_{i,t}(h) < e^{-(K_i(h^{*},h)-\epsilon)t}, ~~\forall~ t \geq t_i(\omega, h , \epsilon). 
	\end{align}
	By Lemma~\ref{lm:bound_true_hypo}, there exists $t'(\omega) \in (0,\infty)$ and $\nu(\omega) >0$ such that $\alpha_{i,t}(h^{*})\geq \nu(\omega), ~\beta_{i,t}(h^{*})\geq \nu(\omega), ~ \forall ~ t \geq t'(\omega), ~\forall ~ i \in V.$ Let $\bar{t}=\max \{t',t_i\}.$ The public belief on $h$ for an agent $i \in D(h^{*},h)$ is:
	\begin{align*}
	\beta_{i,\bar{t}+1}(h) &= \frac{\min\left( \beta_{i,\bar{t}}(h),\left(\hat{\beta}_{j,\bar{t}+1}^{(i)}(h)\right)_{j \in N_i}, \alpha_{i,\bar{t}+1}(h) \right)}{\sum\limits_{k=1}^{M}\min\left( \beta_{i,\bar{t}}(h_k),\left(\hat{\beta}_{j,\bar{t}+1}^{(i)}(h_k)\right)_{j \in N_i}, \alpha_{i,\bar{t}+1}(h_k) \right)} \\
	&\stackrel {\text{(a)}} {\leq} \frac{\alpha_{i,\bar{t}+1}(h)}{\sum\limits_{k=1}^{M}\min\left( \beta_{i,\bar{t}}(h_k),\left(\hat{\beta}_{j,\bar{t}+1}^{(i)}(h_k)\right)_{j \in N_i}, \alpha_{i,\bar{t}+1}(h_k) \right)} \\
	&\stackrel {\text{(b)}} {\leq} \frac{e^{-(K_i(h^{*},h)-\epsilon)(\bar{t}+1)}}{\min\left( \beta_{i,\bar{t}}(h^{*}),\left(\hat{\beta}_{j,\bar{t}+1}^{(i)}(h^{*})\right)_{j \in N_i}, \alpha_{i,\bar{t}+1}(h^{*}) \right)}\\
	&\stackrel {\text{(c)}} {\leq} \frac{e^{-(K_i(h^{*},h)-\epsilon)(\bar{t}+1)}}{\nu}
	\end{align*}
	Here, (a) follows from the definition of minimum, and (b) follows from \eqref{eq:local_upper} and lower bounding the denominator. Step (c) follows from the fact that both local and public beliefs on true hypothesis are bounded away from zero. The same reasoning applies to upper bound $\beta_{i,t}(h)$  for all $t \geq \bar{t}+1.$ Taking the logarithm of this bound and applying the limit, we obtain $-\frac{\log\beta_{i,t}(h)}{t} > (K_i(h^{*},h)-\epsilon)+\frac{\log\nu}{t}.$ By letting $\epsilon$ approach zero, the result follows.
\end{proof}

Now we are ready to prove the Theorem~\ref{th:partial_previous}.

\begin{proof}[Proof of Theorem~\ref{th:partial_previous}]
	For any agent $i \in D(h^{*},h),$ the result directly follows from Lemma \ref{lm:rate_discriminating}. Recall that agents choose a hypothesis $h \in \mathcal{H}$ to communicate at time $t$ with positive probability i.e., $\mathsf{P}(h)>0.$ Thus, by Borel-Cantelli Lemma, each hypothesis is chosen infinitely often for communication. Therefore, there exist infinitely many time steps $t>\hat{t}=\max(t',\bar{t}),$ such that hypothesis $h$ is chosen for communication at time $t.$ Let $j$ be a neighbor of an agent $i \in D(h^{*},h).$ We will now prove a bound on its public belief. Assume that there exists a time $t_1>\hat{t}+1$ such that $h$ is chosen for transmission. i.e., $\tau_{t_1}=h.$ The public belief of an agent $j$ at time $t_1$ is:
	\begin{align*}
	\beta_{j,t_1}(h) &= \frac{\min\left( \beta_{j,t_1-1}(h),\left(\hat{\beta}_{i,t_1}^{(j)}(h)\right)_{i \in N_j}, \alpha_{j,t_1}(h) \right)}{\sum\limits_{k=1}^{M}\min\left( \beta_{j,t_1-1}(h_k),\left(\hat{\beta}_{i,t_1}^{(j)}(h_k)\right)_{i \in N_j}, \alpha_{j,t_1}(h_k) \right)} \\
	&\stackrel {\text{(a)}} {\leq} \frac{\left(\hat{\beta}_{i,t_1}^{(j)}(h)\right)_{i \in N_j}}{\sum\limits_{k=1}^{M}\min\left( \beta_{j,t_1-1}(h_k),\left(\hat{\beta}_{i,t_1}^{(j)}(h_k)\right)_{i \in N_j}, \alpha_{j,t_1}(h_k) \right)} \\
	&\stackrel {\text{(b)}} {\leq} \frac{\beta_{i,t_1-1}(h)}{\left(1-\hat{\beta}_{i,t_1-1}^{(j)}(\tau_{t_1})+\beta_{i,t_1-1}(\tau_{t_1})\right) \left( \min\left( \beta_{j,t_1-1}(h^{*}),\left(\hat{\beta}_{i,t_1}^{(j)}(h^{*})\right)_{i \in N_j}, \alpha_{j,t_1}(h^{*}) \right) \right)}\\
	& \stackrel {\text{(c)}}{\leq}\frac{\beta_{i,t_1-1}(h)}{\nu \left(1-\hat{\beta}_{i,t_1-1}^{(j)}(\tau_{t_1})+\beta_{i,t_1-1}(\tau_{t_1})\right)} \stackrel {\text{(d)}}{=} \frac{\beta_{i,t_1-1}(h)}{\nu \delta}.
	\end{align*}
	Step (a) follows from the definition of minimum. Inequality (b) follows from \eqref{eq:estimate_previous} and the fact that the hypothesis $h$ is chosen for transmission at $t_1.$ Step (c) utilizes the fact that the belief on the true hypothesis is bounded away from zero. Finally (d) follows from the fact that $1-\hat{\beta}_{i,t_1-1}^{(j)}(\tau_{t_1})+\beta_{i,t_1-1}(\tau_{t_1})=\delta > 0.$ We will prove this by contradiction. Let $1-\hat{\beta}_{i,t_1-1}^{(j)}(\tau_{t_1})+\beta_{i,t_1-1}(\tau_{t_1})=0,$ this is true if and only if $\hat{\beta}_{i,t_1-1}^{(j)}(\tau_{t_1})=1$ and $\beta_{i,t_1-1}(\tau_{t_1})=0.$ If $\beta_{i,t_1-1}(\tau_{t_1})=0,$ it implies that $\beta_{j,t_1}(h) \leq 0,$ which is the desired result. If $\beta_{i,t_1-1}(\tau_{t_1})=\epsilon,$ it implies that $1-\hat{\beta}_{i,t_1-1}^{(j)}(\tau_{t_1})+\beta_{i,t_1-1}(\tau_{t_1})=\delta > 0$ because $\hat{\beta}_{i,t_1-1}^{(j)}$ is a probability vector.
	
	Now consider a time $t_2 \geq \hat{t}+1$ such that $\hat{h} \neq h$ is chosen for communication. WLOG, let $t_2=t_1+1,$ then using \eqref{eq:public_previous},
	\begin{equation*}
	\beta_{j,t_2}(h) \leq \frac{\beta_{j,t_2-1}(h)}{\nu}=\frac{\beta_{j,t_1}(h)}{\nu}.
	\end{equation*}
	This proves that $\beta_{j,t}(h)$ goes to zero exponentially fast after $\hat{t}.$
	This completes the argument for $j$ which is a neighbor of $i \in D(h^{*},h).$ Now, consider a neighbor $p$ of $j.$ $\beta_{p,t}(h)$ can be bounded in terms of $\beta_{j,t}(h)$ using similar arguments. As $G$ is strongly connected this argument can be applied to all the agents iteratively thus bounding $\beta_{k,t}(h)~\forall k \in V$ by $\beta_{i,t}.$ The same arguments hold for any $h \in \mathcal{H} \setminus \{h^{*}\}.$ Hence $\beta_{i,t}(h) \xrightarrow{\text{a.s.}}{} 0, \forall h \in \mathcal{H} \setminus \{h^{*}\}$ and $\forall i \in V.$ This implies that $\beta_{i,t}(h^{*}) \xrightarrow{\text{a.s.}}{} 1,~\forall i \in V.$
\end{proof}

Note that the proof of Theorem~\ref{th:partial_previous} works because of the fact that every false hypothesis is shared infinitely often and thus the public beliefs are propagated from discriminating agent to every other agent in $G.$ Thus, if the belief on one fixed hypothesis is shared at all times, the update rules \eqref{eq:estimate_previous}, \eqref{eq:public_previous} will not guarantee true learning in the network.

\begin{corollary} \label{cor:fixed_partial}
	Let the communication graph $G$ be strongly connected, and observation model satisfies Assumption~\ref{as:global}. If every agent shares belief on only one fixed hypothesis $h \in \mathcal{H}$ at all times and uses estimates from \eqref{eq:estimate_previous} to update beliefs in \eqref{eq:local_belief}, \eqref{eq:public_previous}, then true learning does not happen in the network.
\end{corollary}

\subsection{Memory-efficient learning}
\label{sec:memory}

In this section, we present a memory-efficient update rule with partial information sharing. Observe that to compute the estimates $\hat{\beta}_{j,t}^{(i)}$ using \eqref{eq:estimate_previous} agent $i$ needs to maintain $M$ dimensional vector for each of its neighbors\footnote{The memory required at every agent for \eqref{eq:estimate_previous} is proportional to the size of its neighbor set (which could be $n-1$ in the worst case) and the size of hypothesis set.}. In the memory-efficient update \cite{Kayaalp24}, an agent $i$ uses its own belief  to estimate its neighbors' beliefs on missing hypotheses at time $t$ as follows:
\begin{align}
\hat{\beta}_{j,t}^{(i)}(h) = \begin{cases}
\beta_{j,t-1}(h)/C, ~~ h = \tau_t \\
\beta_{i,t-1}(h)/C, ~~ h \neq \tau_t.
\end{cases} \label{eq:estimate_own}
\end{align}
Here $C = 1- \beta_{i,t-1}(\tau_t) + \beta_{j,t-1}(\tau_t)$ to ensure that the estimate $\hat{\beta}_{j,t}^{(i)}$ is a probability vector. Agent $i$ then updates its public belief using \eqref{eq:public_previous}. The following result shows that the true leanring is achieved by all the agents using estimates of \eqref{eq:estimate_own}. 

\begin{theorem}\label{th:partial_own}
	Let the communication graph $G$ be strongly connected, and observation model satisfies Assumption~\ref{as:global}. If every agent shares belief on only one randomly chosen hypothesis at a time and uses estimates from \eqref{eq:estimate_own} to update beliefs in \eqref{eq:local_belief}, \eqref{eq:public_previous}, then true learning happens, i.e., $\beta_{i,t}(h^*) \xrightarrow{\text{a.s.}} 1, ~\forall i \in V.$
\end{theorem}

The proof of Theorem~\ref{th:partial_own} follows on the lines of that of Theorem~\ref{th:partial_previous} presented in Section~\ref{sec:results}. 

\begin{proof}
	For any agent $i \in D(h^{*},h),$ the result directly follows from Lemma \ref{lm:rate_discriminating}. Recall that agents choose a hypothesis $h \in \mathcal{H}$ to communicate at time $t$ with positive probability i.e., $\mathsf{P}(h)>0.$ Thus, by Borel-Cantelli Lemma, each hypothesis is chosen infinitely often for communication. Therefore, there exist infinitely many time steps $t>\hat{t}=\max(t',\bar{t}),$ such that hypothesis $h$ is chosen for communication at time $t.$ Let $j$ be a neighbor of an agent $i \in D(h^{*},h).$ We will now prove a bound on its public belief. Assume that there exists a time $t_1>\hat{t}+1$ such that $h$ is chosen for transmission. i.e., $\tau_{t_1}=h.$ The public belief of an agent $j$ at time $t_1$ is:
	\begin{align*}
	\beta_{j,t_1}(h) &= \frac{\min\left( \beta_{j,t_1-1}(h),\left(\hat{\beta}_{i,t_1}^{(j)}(h)\right)_{i \in N_j}, \alpha_{j,t_1}(h) \right)}{\sum\limits_{k=1}^{M}\min\left( \beta_{j,t_1-1}(h_k),\left(\hat{\beta}_{i,t_1}^{(j)}(h_k)\right)_{i \in N_j}, \alpha_{j,t_1}(h_k) \right)} \\
	&\stackrel {\text{(a)}} {\leq} \frac{\left(\hat{\beta}_{i,t_1}^{(j)}(h)\right)_{i \in N_j}}{\sum\limits_{k=1}^{M}\min\left( \beta_{j,t_1-1}(h_k),\left(\hat{\beta}_{i,t_1}^{(j)}(h_k)\right)_{i \in N_j}, \alpha_{j,t_1}(h_k) \right)} \\
	&\stackrel {\text{(b)}} {\leq} \frac{\beta_{i,t_1-1}(h)}{\left(1-\beta_{j,t-1}(\tau_t)+\beta_{i,t-1}(\tau_t)\right) \left( \min\left( \beta_{j,t_1-1}(h^{*}),\left(\hat{\beta}_{i,t_1}^{(j)}(h^{*})\right)_{i \in N_j}, \alpha_{j,t_1}(h^{*}) \right) \right)}\\
	& \stackrel {\text{(c)}}{\leq}\frac{\beta_{i,t_1-1}(h)}{\nu \left(1-\beta_{j,t-1}(\tau_t)+\beta_{i,t-1}(\tau_t)\right)} \stackrel {\text{(d)}}{=} \frac{\beta_{i,t_1-1}(h)}{\nu \delta}.
	\end{align*}
	Step (a) follows from the definition of minimum. Inequality (b) follows from \eqref{eq:estimate_own} and the fact that the hypothesis $h$ is chosen for transmission at $t_1.$ Step (c) utilizes the fact that the belief on the true hypothesis is bounded away from zero. Finally (d) follows from the fact that $1-\beta_{j,t-1}(\tau_t)+\beta_{i,t-1}(\tau_t)=\delta > 0.$ We will prove this by contradiction. Let $1-\beta_{j,t-1}(\tau_t)+\beta_{i,t-1}(\tau_t)=0,$ this is true if and only if $\beta_{j,t-1}(\tau_t)=1$ and $\beta_{i,t_1-1}(\tau_{t_1})=0.$ If $\beta_{i,t_1-1}(\tau_{t_1})=0,$ it implies that $\beta_{j,t_1}(h) \leq 0,$ which is the desired result. If $\beta_{i,t_1-1}(\tau_{t_1})=\epsilon,$ it implies that $1-\beta_{j,t-1}(\tau_t)+\beta_{i,t-1}(\tau_t)=\delta > 0$ because $\beta_{j,t-1}$ is a probability vector.
	
	Now consider a time $t_2 \geq \hat{t}+1$ such that $\hat{h} \neq h$ is chosen for communication. WLOG, let $t_2=t_1+1,$ then using \eqref{eq:public_previous},
	\begin{equation*}
	\beta_{j,t_2}(h) \leq \frac{\beta_{j,t_2-1}(h)}{\nu}=\frac{\beta_{j,t_1}(h)}{\nu}.
	\end{equation*}
	This proves that $\beta_{j,t}(h)$ goes to zero exponentially fast after $\hat{t}.$
	This completes the argument for $j$ which is a neighbor of $i \in D(h^{*},h).$ Now, consider a neighbor $p$ of $j.$ $\beta_{p,t}(h)$ can be bounded in terms of $\beta_{j,t}(h)$ using similar arguments. The same arguments hold for any $h \in \mathcal{H} \setminus \{h^{*}\}.$ Hence $\beta_{i,t}(h) \xrightarrow{\text{a.s.}}{} 0, \forall h \in \mathcal{H} \setminus \{h^{*}\}$ and $\forall i \in V.$ This implies that $\beta_{i,t}(h^{*}) \xrightarrow{\text{a.s.}}{} 1,~\forall i \in V.$
\end{proof}

\section{Numerical simulations}
\label{sec:simulation}

We consider a $100$ agents $4$-regular undirected strongly connected network. Each agent has $4$ neighbors. Observation set $\mathcal{O}_i$ of every agent $i$ has $500$ distinct signals and number of possible hypotheses are taken as $|\mathcal{H}| =20.$ For each agent the likelihood functions $f_i(o_i|h)~\forall h \in \mathcal{H}, \forall o_i \in \mathcal{O}_i$ are generated randomly and we assume $h_1$ to be the true hypothesis. We generate likelihood functions for agent $1$ such that it can distinguish all the hypotheses. Figure~\ref{fig:compare}a shows the evolution of a typical agent's belief on true hypothesis under all the update rules. Observe that agents learn the true hypothesis under all the three update models presented in Section~\ref{sec:model} and Section~\ref{sec:partial}. Figure~\ref{fig:compare}b shows the rate of convergence of the belief on a false hypothesis for an agent in the network. Observe that agents converge fastest when full information is shared (see \eqref{eq:public_full}), followed by convergence with partial information with estimates using previous beliefs (see \eqref{eq:estimate_previous}). As expected the convergence is the slowest for memory-efficient update (see \eqref{eq:estimate_own}) because of the use of agent's own beliefs to estimate the missing beliefs at any time.

\begin{figure}
	\subfigure[]{%
		\includegraphics[width=0.45\linewidth]{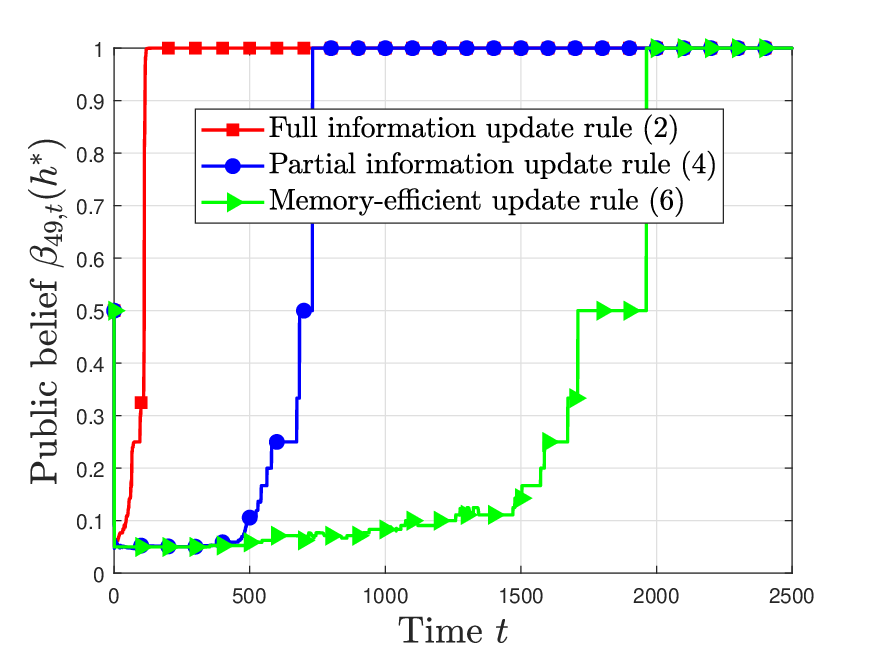}
	}
	\quad
	\subfigure[]{%
		\includegraphics[width=0.45\linewidth]{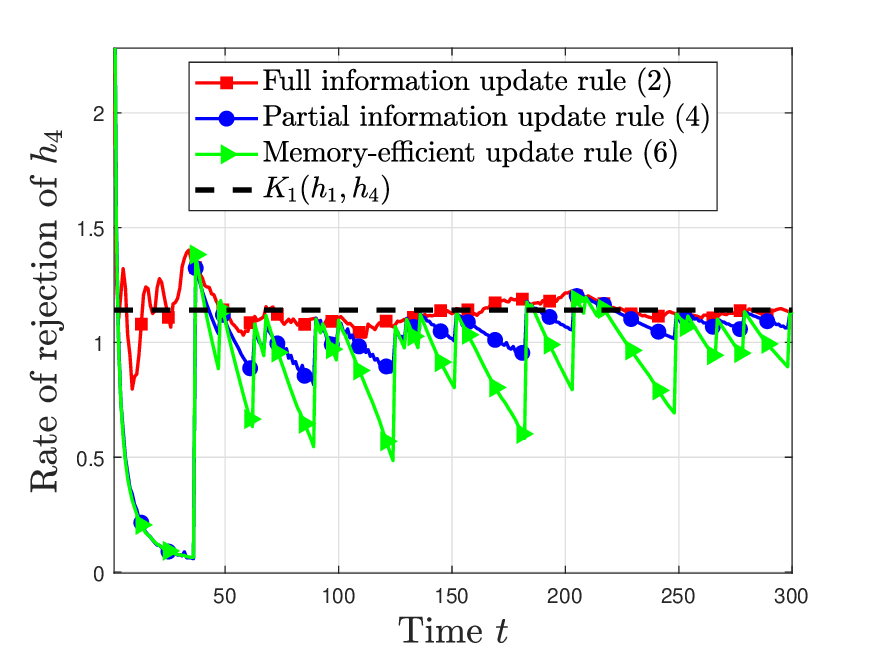}
	}
	
	\caption{Simulation results for a $100$ agents $4$-regular network using various update rules. (a) Evolution of $\beta_{i,t}(h^*)$ for a typical non-discriminating agent. (b) Rate of rejection of a false hypothesis, $h_4,$ namely $r_{i,t}(h_4) = -\frac{\log \beta_{i,t}(h_4)}{t}$ for a non-discriminating agent.}
	\label{fig:compare}
\end{figure}

\section{Conclusion}
\label{sec:conclusion}

In this work, we presented two update rules which enable agents to share belief on only one hypothesis at a time using min-rule for distributed learning. In the first rule, agents estimate missing beliefs of neighbors by storing previously shared beliefs and in the second rule (memory-efficient) they estimate using their own beliefs. We show that true learning happens almost surely in both cases and present simulation results to compare the rate of convergence of the proposed update rules. Future direction of work is to analyze the effect of quantization in distributed learning with partial information sharing.

\bibliographystyle{IEEEtran}
\bibliography{DL_QL_bib}

\end{document}